\documentclass[aps,reprint,amsfonts, amssymb, amsmath,pra, superscriptaddress, twocolumn,longbibliography,nofootinbib]{revtex4-1}

\usepackage{xcolor}

\usepackage{float}

\usepackage[shortlabels]{enumitem}

\usepackage{braket}
\usepackage{amsthm}
\usepackage{mathtools}
\usepackage{url}
\usepackage{physics}
\usepackage{graphicx}
\usepackage[left=16mm,right=16mm,top=35mm,columnsep=15pt]{geometry} 

\usepackage[T1]{fontenc}

\usepackage{bm}

\usepackage{silence}
\newcommand*{\eh}{\mathrm{End\, }(\mathcal{H})}

\def\ad{^{\dagger}}

\newcommand{\fsnull}[1]{}
\newcommand{\old}[1]{}

\addtocontents{toc}{\protect\setcounter{tocdepth}{1}}

\usepackage{hyperref}
\usepackage[toc,page,header]{appendix}

\addtocontents{toc}{\protect\setcounter{tocdepth}{1}}

\usepackage[makeroom]{cancel}
\definecolor{C1}{RGB}{52, 89, 149}
\definecolor{C2}{RGB}{251, 77, 61}
\definecolor{C3}{RGB}{3, 206, 164}
\definecolor{C4}{RGB}{202, 21, 81}
\definecolor{C5}{RGB}{202, 21, 81}
\definecolor{C6}{RGB}{202, 51, 121}
\hypersetup{colorlinks=true, linkcolor=C6, citecolor=C6, urlcolor=C6}

\usepackage{tikz}
\tikzset{every picture/.style=remember picture}

\usepackage[utf8]{inputenc}
\usepackage{graphicx}
\usepackage{xcolor}
\usepackage{amsmath}
\usepackage{amsthm}
\usepackage{bm}
\usepackage{bbm}
\usepackage{comment}
\usepackage{mathdots}
\usepackage{lipsum}
\usepackage{verbatim}
\usepackage{natbib}
\usepackage{nccmath}
\usepackage{amsfonts}
\usepackage{thm-restate}
\usepackage{thmtools}
\usepackage{ytableau}

\usepackage{amssymb}
\usepackage{dsfont}

\newcommand{\Var}{{\rm Var}}

\renewcommand{\geq}{\geqslant}
\renewcommand{\leq}{\leqslant}

\newcommand{\bs}{\textsf{BS}}

\newcommand{\mcl}{\mathcal{L}}

\newcommand{\mcf}{\mathcal{F}}
\newcommand{\mco}{\mathcal{O}}

\newcommand{\mch}{\mathcal{H}}
\newcommand{\mcm}{\mathcal{M}}

\newcommand{\mcp}{\mathcal{P}}
\newcommand{\mcd}{\mathcal{D}}
\newcommand{\mce}{\mathcal{E}}

\newcommand{\mcz}{\mathcal{Z}}
\newcommand{\mcr}{\mathcal{R}}

\newcommand{\mbc}{\mathbb{C}}
\newcommand{\mbr}{\mathbb{R}}

\newcommand{\mbe}{\mathbb{E}}

\def\be{\begin{equation}}
\def\ee{\end{equation}}
\def\bs{\begin{split}}
\def\e{\end{split}}
\def\ba{\begin{eqnarray}}
\def\bea{\begin{eqnarray}}

\def\tea{\end{eqnarray}}
\def\ea{\end{eqnarray}}
\def\eea{\end{eqnarray}}

\def\g{\mathfrak{g}}

\newcommand{\id}{\mathds{1}}

\newcommand{\sbraket}[2]{ \langle#1 | #2 \rangle}

\addtocontents{toc}{\protect\setcounter{tocdepth}{0}}

\def\mg{\mathsf{MG}(n)}

\def\sh{{\rm sh}}

\def\g{\gamma}

\DeclareMathOperator*{\expect}{\mathbb{E}}

\def\be{\begin{equation}}
\def\te{\end{equation}}
\def\ee{\end{equation}}
\def\ba{\begin{eqnarray}}
\def\bea{\begin{eqnarray}}

\def\tea{\end{eqnarray}}
\def\ea{\end{eqnarray}}
\def\eea{\end{eqnarray}}

\begin{document}

\makeatletter
\newif\ifinappendixtoc  

\newcommand{\tableofcontentsappendixonly}{%
  \begingroup
    \inappendixtocfalse

    \@ifundefined{l@section}{}{%
      \let\ao@l@section\l@section
      \def\l@section##1##2{\ifinappendixtoc \ao@l@section{##1}{##2}\fi}%
    }
    \@ifundefined{l@subsection}{}{%
      \let\ao@l@subsection\l@subsection
      \def\l@subsection##1##2{\ifinappendixtoc \ao@l@subsection{##1}{##2}\fi}%
    }
    \@ifundefined{l@subsubsection}{}{%
      \let\ao@l@subsubsection\l@subsubsection
      \def\l@subsubsection##1##2{\ifinappendixtoc \ao@l@subsubsection{##1}{##2}\fi}%
    }

    \def\AppendixTOCMark{\global\inappendixtoctrue}%

    \tableofcontents
  \endgroup
}
\makeatother

\title{
Practical fermionic shadows enabled by improved sample-complexity bounds 
}

\author{Maxwell West}
\thanks{westm@lanl.gov}
\affiliation{Theoretical Division, Los Alamos National Laboratory, Los Alamos, New Mexico 87545, USA}
\affiliation{Quantum Science Center, Oak Ridge, TN 37931, USA}

\author{Su Yeon Chang}
\affiliation{Theoretical Division, Los Alamos National Laboratory, Los Alamos, New Mexico 87545, USA}

\author{Luke Coffman}
\affiliation{Theoretical Division, Los Alamos National Laboratory, Los Alamos, New Mexico 87545, USA}
\affiliation{Department of Physics, Harvard University, Cambridge, MA 02138, USA}
\affiliation{School of Engineering and Applied Sciences, Harvard University, Cambridge, MA 02138, USA}

\author{Mart\'{i}n Larocca}
\affiliation{Theoretical Division, Los Alamos National Laboratory, Los Alamos, New Mexico 87545, USA}
\affiliation{Quantum Science Center, Oak Ridge, TN 37931, USA}

\author{M. Cerezo}
\affiliation{Information Sciences, Los Alamos National Laboratory, Los Alamos, New Mexico 87545, USA}
\affiliation{Quantum Science Center, Oak Ridge, TN 37931, USA}

\begin{abstract}
Classical shadow tomography is widely touted as supplying a family of methods for extracting information from quantum systems with polynomially scaling sample-complexities. In the current era of quantum computers possessing on the order of hundreds of qubits, however, polynomial scaling can nonetheless be  prohibitive. Thus, there is a strong practical need for obtaining sample-complexity bounds which are as tight as possible. Here we address this in the case of fermionic (matchgate) shadows. For an arbitrary observable $O$ of Majorana degree $2k$, we improve the previously known sample-complexity bound of $\mathcal{O}(n^{2k}\|O\|_\infty^2)$ to $\mathcal{O}(n^{k}\|O\|_\infty^2)$, which is asymptotically tight. For example, for the estimation of the energy per mode of an open  fermionic 50-site Hubbard chain with hopping and on site strenghts respectively given by $t=1$, $V=4$, and for a target additive precision of 0.1, this   reduces the  number of required shots from $\sim 10^9$   to $\sim  10^5$. That is, the new bound reduces the required number of shots by approximately $99.98\%$.
\end{abstract}

\maketitle


\section{Introduction}
State-of-the-art programmable quantum computers are becoming increasingly large. On the one hand, this has led to a sequence of increasingly plausible (if   not yet universally regarded as entirely unimpeachable) claims of quantum advantage in the form of  sampling from various distributions, some already several years old~\cite{arute2019quantum,morvan2023phase,google2025observation,hangleiter2023computational}. On another, however, the ability of current generation quantum computers to return (say) energy estimates of molecules  of relevance to modern quantum chemistry in an advantageous fashion remains more contested~\cite{hartnett2026fast}. Towards the goal of usefully extracting such expectation values from quantum computers, a leading framework is that of \textit{classical shadows}~\cite{huang2020predicting,aaronson2019shadow,west2026classical,bertoni2024shallow,low2022classical,zhao2021fermionic,wan2022matchgate,west2026particle,van2022hardware,zhao2024group,king2024triply,jerbi2023shadows,koh2022classical,chen2021robust,hearth2024efficient,chan2022algorithmic,sauvage2024classical,helsen2023thrifty,kunjummen2023shadow,west2024random,grier2024sample,brandao2020fast,bertoni2024shallow,vitale2024estimation,somma2024shadow,west2025real,bringewatt2025classical}.  \\

Classical shadows refers to a broad class of protocols~\cite{west2026classical} of significantly varying experimental complexity. For example, local Clifford shadows~\cite{huang2020predicting} are of enviable scalability, and have now been experimentally demonstrated on systems of nearly 100 qubits~\cite{fischer2025large,votto2026learning}. Contrarily, and despite   the initial optimism surrounding their potential applications to problems in quantum chemistry,   experimental  implementations of free-fermionic (matchgate) shadows~\cite{wan2022matchgate,zhao2021fermionic,low2022classical,west2026particle} have lagged somewhat behind, with demonstrations restricted to only 16 qubits~\cite{zhao2026quantum}. This is not inexplicable, however. In particular, there are two primary axes along which matchgate shadows demand considerably greater complexity than other common shadow protocols. The first is the circuit depth required to implement the randomised measurements entailed by each protocol. On a one-dimensional circuit architecture with nearest-neighbour interactions, these depths are constant for local Cliffords~\cite{huang2020predicting}, logarithmic for global Cliffords~\cite{schuster2024random}, and linear\footnote{Although we do note that, in first quantisation, the closely related particle-preserving fermionic shadows can be implemented in log depth~\cite{west2026particle}.} for matchgate shadows~\cite{west2025no,grevink2025will,west2026ambient,braccia2025optimal}. A second source of increased difficulty in the matchgate case concerns the number of shots (i.e.,  sample-complexity) required to estimate natural classes of observables as a function of the system size. Indeed, whereas observables of constant spatial locality lead, in the local Clifford case, to sample-complexities that are independent of the system size~\cite{huang2020predicting}, operators of constant \textit{fermionic locality} lead, in the matchgate case, to polynomially growing sample-complexities~\cite{wan2022matchgate,zhao2021fermionic}.
\\

In this work, we thus tackle the question of the optimal sample-complexity of matchgate classical shadows. For  the estimation of a given \textit{Majorana monomial} of degree $2k$ (seeing Section~\ref{sec:preliminaries} for definitions) corresponding to an $n$-mode system, it is known~\cite{wan2022matchgate,zhao2021fermionic,west2026classical} that this sample-complexity scales exactly as $ \sim n^{k}$. For a more general fermionic operator $O$ of degree $2k$, however, the previously known~\cite{west2026classical} sample-complexity bounds were of order $\mco(n^{2k}\|O\|_\infty^2)$. Our   contribution here is to improve this bound to $\mco(n^{k}\|O\|_\infty^2)$, bringing it in line with the optimal scaling of a single Majorana monomial. While this is ``only'' a polynomial improvement, we give some numerical estimates to show that this saved factor   of $n^k$ can be highly beneficial in experimentally relevant regimes (see Figs.~\ref{fig:schematic} and~\ref{fig:scaling}).

\section{Preliminaries} \label{sec:preliminaries}
Let us begin by establishing some notation. We will   denote by   $  \mch_n=(\mathbb C^2)^{\otimes n} $ the Hilbert space of $n$ fermionic modes, and let $\gamma_1,\ldots,\gamma_{2n}$ be  Majorana operators satisfying the canonical anti-commutation relations $\{\gamma_\mu,\gamma_\nu\}=2\delta_{\mu\nu}\id$; 
we will  make the explicit choice~\cite{wan2022matchgate}
\begin{align}\label{eq:majos}
\gamma_1&=XIII\cdots I\;,&\gamma_{2}&=YIII\cdots I\nonumber\\
\gamma_3&=ZXII\cdots I\;,&\gamma_{4}&=ZYII\cdots I\\
&\:\:\vdots&&\:\:\vdots\nonumber\\
\gamma_{2n-1}&=ZZ\cdots ZX\;,&\gamma_{2n}&=ZZ\cdots ZY\;.  \nonumber
\end{align}
For a  set $S=\{s_1<\cdots<s_{k}\}\subseteq[2n]$, we define the Majorana monomial
\begin{equation}
  \Gamma_S:=(-i)^{k(k-1)/2}\gamma_{s_1}\cdots\gamma_{s_{k}}\,,
\end{equation}
where the factors of $i$ are included to render $\Gamma_S$ hermitian. In particular, we write $\Pi:=\Gamma_{[2n]}$ for the fermionic parity operator. We will be particularly interested in the   group of fermionic Gaussian unitaries (the so-called matchgate group~\cite{jozsa2008matchgates}), whose elements are generated by quadratic combinations of the Majoranas:
\begin{equation}
    \mg = \Big\{ \exp\Big( \sum_{1\le \mu<\nu\le 2n} h_{\mu,\nu}\g_\mu \g_\nu\Big)\ |\ h_{\mu,\nu}\in\mbr\Big\}\,.
\end{equation}
Under the standard action of the matchgate group on $\mch$, and its adjoint action on $\mcl:=\eh$, we have the well-known decompositions~\cite{west2026classical}
\begin{equation}\label{eq:opspace_decomp}
    \mch \cong \mch^+ \oplus \mch^-, \quad \mcl \cong \bigoplus_{k=0}^{2n} \mcl_{k}\,,
\end{equation}
where $\mch^\pm$ are the subspaces of $\mch$ spanned by the computational basis strings of even and odd parity, and
\begin{equation}
  \mcl_{k}:=\operatorname{span}_{\mathbb R}\{\Gamma_S:S\subseteq[2n],\ |S|=k\}
\end{equation}
is what we will call the \textit{homogeneous sector} of order $k$. The above spaces are  all irreps of $\mg$, with the exception of the reduction $\mcl_n\cong\mcl_n^+\oplus \mcl_n^-$ into $\pm 1$ eigenspaces of the parity operator $\Pi$. \\

\begin{figure}
    \centering
    \includegraphics[width=.8\columnwidth]{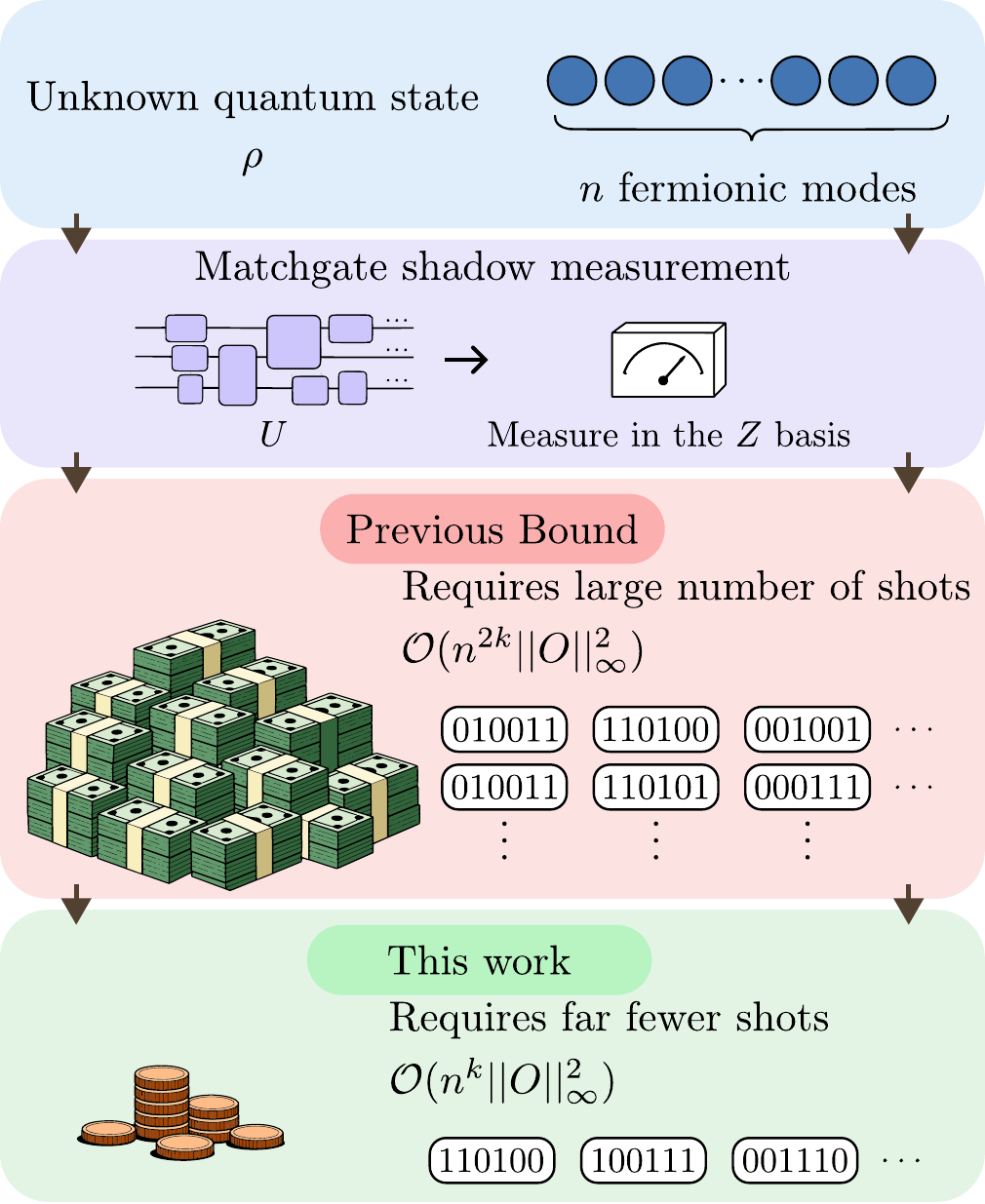}
    \caption{Schematic illustration of the improvement in sample-complexity obtained in this work. An unknown state of $n$ fermionic modes is measured using matchgate classical shadows. Previous bounds require $\mathcal{O}(n^{2k}\|O\|_\infty^2)$ shots to estimate an observable $O$ of Majorana degree $2k$, whereas our results reduce this requirement to $\mathcal{O}(n^k\|O\|_\infty^2)$, yielding a substantial reduction in experimental cost.}
    \label{fig:schematic}
\end{figure}

Next  we recall that,  given a measurement basis $\mcz$ and an ensemble $\mce$ of unitaries, we have an induced \textit{classical shadow protocol}~\cite{huang2020predicting}  for extracting information from an unknown quantum state $\rho$.   In this  protocol one   evolves copies of $\rho$ under the action of unitaries sampled from $\mce$, and measures the resulting state with respect to $\mcz$.  Conditioned  on sampling some unitary $U$ and measuring some outcome $\ket z$, one records the state $U \ad\ketbra{z}{z} U$. This  effects the \textit{measurement channel}\footnote{Representation theoretically, this map is actually very natural; with $\mcd_\mcz$  the dephasing channel wrt $\mcz$, Eq.~\eqref{eq:chan} is simply the (average) of the natural action of   the group of unitaries on $\mcd_\mcz\in {\rm End}({\rm End}(\mch))$; that is, we have $\mcm = \expect_{u\sim G} u \cdot \mcd_\mcz$.}
\begin{equation}\label{eq:chan}
  \mcm(\rho):=\expect_{U\sim\mu_G}\sum_{z\in\mcz}
    \langle z|U\rho U^\dagger|z\rangle
    U^\dagger|z\rangle\!\langle z|U\,,
\end{equation}
where we have taken $\mce$ to form a 3-design over a compact group $G$, with Haar measure  $\mu_G$. In the case of $G$ the matchgate group and $\mcz$ the computational basis (which we will henceforth assume) the measurement channel is diagonal with respect to the decomposition of Eq.~\eqref{eq:opspace_decomp}. Indeed, one has that for any $X\in\mcl_{2k}$~\cite{wan2022matchgate,zhao2021fermionic,west2026classical},
\begin{equation}
  \mcm(X)=a_{n,k}X,\qquad a_{n,k}:= \binom{n}{k} \big/ \binom{2n}{2k} \,,
  \label{eq:visibility}
\end{equation}
and, for $X\in\mcl_{2k+1}$, $\mcm(X)=0$.
Now, from the non-vanishing of $a_{n,k}$, it follows from   general representation theoretic considerations~\cite{west2026classical} that for $O\in\mcl_{2k}$, the one-shot estimator associated with the outcome $(U,z)$ given by  
\begin{equation}
  \hat o=\Tr\!\left[O\mcm^{-1}(U\ad \ketbra{z}{z} U)\right]
  =a_{n,k}^{-1}\Tr(OU\ad \ketbra{z}{z} U)
\end{equation}
is unbiased for $\Tr(O\rho)$.  The number of single-shot estimators $\mathsf{S}$ required to ensure, with probability $1 - \delta$, that the empirical average of each observable in a collection $\{O_i\}_{i=1}^M$ is within an additive error $\varepsilon$ of its true expectation value is given by
by~\cite{huang2020predicting}  
\begin{equation}\label{eq:sc}
\mathsf{S} \in \mco\left(\frac{\log (M/\delta)}{\varepsilon^2}\max_i {\rm Var}\, [\hat{o}_i ]\right)\,.
\end{equation}
 For an observable $O$ of interest, then, one wishes to bound the so-called \textit{shadow norm}~\cite{huang2020predicting}  
\begin{align}\label{eq:sn}
  \|O\|_{{\rm sh}}^2=&\sup_{ \rho}\! \underset{U\sim\mu_G}{\int}\sum_{z\in \mcz} \sbraket{z|U\rho U\ad}{z} \Tr\!\left[O\mcm^{-1}(U\ad\! \ketbra{z}{z}\! U)\right]^2,
\end{align}
which itself bounds the variance $\Var[\hat{o}_i]$ for every state. Naturally, quite a lot of effort has gone into bounding these shadow norms for various classes of observables. For example, when the target observable is a \emph{single} Majorana string of degree $2k$, one finds that the variance is bounded by $a_{n,k}^{-1}$~\cite{wan2022matchgate,zhao2021fermionic,west2026classical}.  More generally, for an observable $O=\sum_{k}O_{2k}\,,$ where $O_{2k}$ is the component of the observable in the subspaces $\mcl_{2k}$ of Eq.~\eqref{eq:opspace_decomp}, one has on general grounds~\cite{west2026classical} that, for any $\rho$,
\begin{equation}
\Var[\hat{o}] \leq \sum_{k} \frac{\|O_{2k}\|^2_2}{a_{n,k}}  
\label{eq:2nb}         
\end{equation}
and additionally
\begin{equation}
\Var[\hat{o}] \leq  \left\|   \sum_{k} \frac{O_{2k} }{a_{n,k}}  \right\|_\infty^2. \label{eq:inb}
\end{equation}
Unfortunately, the bounds of Eqs.~\eqref{eq:2nb} and~\eqref{eq:inb} can be quite loose. For example, one is often working with observables of extensive 2-norm (e.g., Pauli strings), so that the bound of Eq.~\eqref{eq:2nb} is immediately impractical. The infinity norm   bound of Eq.~\eqref{eq:inb}, on the other hand, is penalised by an additional factor of $a_{n,k}$. One would ideally like a situation with the best of both worlds. That is, a bound involving the infinity norm and only one factor of $a_{n,k}$. Our main result is to derive exactly such a bound in the case of matchgate shadows.

\section{Results} \label{sec:results}

Our main result is a theorem that bounds the the variance for an observable that is fully supported on some $\mcl_{2k}$, i.e., which is expressed as a homogeneous combination of Majoranas of degree $2k$. This result will be then generalized for more general operators $O$. 

\begin{restatable}{theorem}{homogeneousbound}\label{thm:homogeneous}
Let $n\geq1$,  and $O\in\mcl_{2k}$ for some  $0\leq k\leq n/2$.  Then the variance of the estimator obtained by matchgate shadows satisfies 
\begin{equation}
  \Var[\hat o] \leq\frac{3}{2a_{n,k}}\|O\|_\infty^2\, .
\end{equation}
Up to the factor of $3/2$, this bound is optimal.
\end{restatable}

Here we present a brief sketch of the proof. The full derivation can be found in the appendices.

\begin{proof}[Proof of Theorem~\ref{thm:homogeneous} (sketch)]
Our first step is to exploit the fact that the Clifford matchgates form a matchgate 3-design~\cite{wan2022matchgate,gargiulo2026pauli}, so that we may replace the matchgate Haar average of Eq.~\eqref{eq:sn} with a finite average over the  Clifford matchgates. The usefulness of this is, as we shall see, its reduction of a considerable part of the problem to combinatorics.   \\ 

To see this, note that if we were to omit the action of the random matchgate Clifford from the shadow protocol, we would simply be measuring the unknown state in the computational basis; that is, the simultaneous eigenbasis of the $Z_j=-i\g_{2j-1}\g_{2j}$. As Clifford matchgates act on Majoranas as signed permutations, in a given shot we effectively measure in the simultaneous eigenbasis of the commuting Majorana monomials associated with a uniformly random perfect matching $M$ of the $2n$ Majorana indices. For example, take $n=2$. The computational basis is the simultaneous eigenbasis of $Z_1\propto \g_1\g_2$ and $Z_2\propto \g_3\g_4$, and thus corresponds to the perfect matching $\{\{1,2\},\{3,4\}\}$. Under the action of the matchgate Clifford which (say) fixes $\g_1$ and $\g_4$, and permutes $\g_2$ and $\g_3$, we would have the perfect matching $\{\{1,3\},\{2,4\}\}$, which corresponds to measuring in the simultaneous eigenbasis of $YX\propto \g_1\g_3$ and $XY\propto \g_2\g_4$. \\

Now, let $\mcd_M$ be the dephasing channel with respect to  the common eigenbasis of the $n$ commuting
Majorana monomials associated with the perfect matching $M$ in the above-prescribed fashion. We then have $\expect_\rho[\hat o^{2}]=\Tr(\rho B_O)$, where

\begin{equation}
  B_O=a_{n,k}^{-2}\mbe_M[\mcd_M(O)^2]\,.
  \label{eq:ssm}
\end{equation}
In particular, $B_O\geq0$ and
$\|O\|_{\sh}^2=\|B_O\|_\infty$.  
Let us expand $O=\sum_{|S|=2k}c_S\Gamma_S$, so that  
\begin{equation}
    B_O = a_{n,k}^{-2}\mbe_M\, \sum_{S,T}c_Sc_T \mcd_M(\Gamma_S)\mcd_M(\Gamma_T)\,.
\end{equation}
A given monomial $\Gamma_S$ is diagonal in the basis  specified by $M$ (and so survives $\mcd_M$) precisely when $i\in S\Leftrightarrow M(i)\in S$, where $M(i)$ is the pair of the index $i$ specified by   $M$. Otherwise  it is annihilated. This yields 
\begin{equation}
    B_O = a_{n,k}^{-2}  \sum_{S,T}c_Sc_T p_{S, T} \Gamma_S \Gamma_T\,,
\end{equation}
with $p_{S, T}$ the probability that both $\Gamma_S$ and $\Gamma_T$ survive under a random perfect matching.
Some counting shows that if
$|S\cap T|=2j$, this probability   is
\begin{equation}\label{eq:prob}
  p_j=\frac{(2j-1)!!\,(2(k-j)-1)!!^2(2(n-2k+j)-1)!!}{(2n-1)!!}\,;
\end{equation}
if their intersection is odd, then the probability of their joint survival is zero (see Appendix~\ref{sec:matching}). At this point one can show the surprising result that if $|S\cap T|$ is odd, then $\Gamma_S$ and $\Gamma_T$ anticommute, and therefore cancel in $O^2$. Hence,we can equivalently write\footnote{Note  that if $O^2=\id$ (as, for example, in the case of a single Majorana string) then Eq.~\eqref{eq:sk} immediately yields  $\|O\|^2_{\rm sh}=a_{n,k}^{-1}$}
\begin{equation}
  B_O=a_{n,k}^{-2}\sum_{r=0}^{k}p_{k-r}\mcp_{4r}(O^2)=a_{n,k}^{-1}\sum_{r=0}^{k}q_{r}\mcp_{4r}(O^2)\,, \label{eq:sk}
\end{equation}
where $\mcp_{4r}(O^2)$ denotes the degree-$4r$ component of $O^2$ (note that for an intersection of size $2(k-r)$ for some $r$, the product $\Gamma_S \Gamma_T$ has degree $4r$), and we have  introduced $q_r=p_{k-r}/a_{n,k}$. \\

So, let us then attempt to control the operator norm of $\sum_rq_r\mcp_{4r}(O^2)$. First, let $V_\mu=i\Pi\g_\mu$ be the (Hermitian) unitary
whose conjugation action flips the sign of $\g_\mu$ and fixes all the other Majoranas (recall that $\Pi=\Gamma_{[2n]}$ is the parity operator). 
Then, for $A\subseteq[2n]$, let $\mathcal U_A$ be the composition of the conjugation channels of the Majoranas whose indices lie within $A$, and for $ 0\leq j\leq N$ define the \textit{radial shell channel}

\begin{equation}
  \mathcal R_j:=\binom {2n}{j}^{-1}\sum_{\substack{A\subseteq[2n]\\|A|=j}}\mathcal U_A\,.
\end{equation}
Every $\mathcal R_j$ is unital and completely positive, and has eigenvalues given by normalized \textit{Krawtchouk polynomials} $\kappa_s^{(2n)}(j)$, with corresponding  eigenspaces given by the span of the degree $s$ Majoranas  (see Appendix~\ref{sec:matching}). The point of introducing these objects is that we can find (see Appendix~\ref{sec:kt}) a function $w:[2n]\to\mbr$ satisfying $\sum_{j\,:\,w_j>0} w_j\leq3/2$, and, critically, 
\begin{align}
  \sum_{j=0}^{2n}w_j\kappa_{4r}^{(2n)}(j)&=q_r\,.
\end{align}
That is, the \textit{Krawtchouk transform} of $w$ evaluated at the values $4r$ gives exactly the sequence $q_r$ of Eq.~\eqref{eq:sk}. Now, letting $R=\|O\|_\infty^2$, we have $0\leq O^2\leq R\id$, with each $\mcr_j$ preserving this ordering (by their unitality and positivity); in particular,
we have
\begin{align*}
    0&\leq a_{n,k} B_O\\ 
    &=\sum_{j=0}^{2n}w_j \mcr_{j}(O^2) \leq \sum_{j\,:\, w_j\geq 0} w_j \mcr_{j}(O^2) \leq \frac{3\|O\|_\infty^2}{2}\,.
\end{align*}
Putting everything together we therefore have
\begin{align*}
    \|O\|_{\sh}^2&= \|B_O\|_\infty\\
    &=a_{n,k}^{-1}\Big\|\sum_{r=0}^{k}q_{r}\mcp_{4r}(O^2)\Big\|_\infty\\
    &=a_{n,k}^{-1}\Big\|\sum_{r=0}^{k}\Big(\sum_{j=0}^{2n}w_j\kappa_{4r}^{(2n)}(j)\Big)\mcp_{4r}(O^2)\Big\|_\infty\\
    &= a_{n,k}^{-1}\Big\|\sum_{j=0}^{2n} w_j \mcr_{j}(O^2)\Big\|_\infty\\
    &\leq \frac{3}{2a_{n,k}}\|O\|_\infty^2\,.
\end{align*}
As we shall see in the appendices, the most difficult bit is actually constructing a function $w$ with the stated conditions; modulo that (and the omitted combinatorics from above) the proof of our bound is complete. Finally, its optimality may be seen from the fact that (up to the factor of 3/2) it is saturated when the target observable is a single Majorana string of degree $2k$~\cite{wan2022matchgate,zhao2021fermionic,west2026classical}
\end{proof}

As previously noted, Theorem~\ref{thm:homogeneous} covers the case of a  homogeneous observable $O\in\mcl_{2k}$ for some $k$. In the more general case where $O$ lives in multiple subspaces, we have
\begin{restatable}{crl}{crlmult}\label{crl:mult}
   If $O=\sum_{k=1}^K O_{2k}$, with $O_{2k}\in \mcl_{2k}$ for all $k$, then
   \begin{equation}\label{eq:mult1}
       \Var(\hat o) \leq \frac32 \left(\sum_{k=1}^K\frac{\|O_{2k}\|_\infty}{\sqrt{a_{n,k}}}\right)^2\,;
   \end{equation}
   if $K\leq n/2$ we further have
   \begin{equation}\label{eq:mult2}
       \Var(\hat o) \leq \frac{3K}{2 a_{n,K}}\sum_{k=1}^K \|O_{2k}\|_\infty ^2\,.
   \end{equation}
\end{restatable}
Corollary~\ref{crl:mult}  follows quickly from Theorem~\ref{thm:homogeneous} by the triangle inequality (see Appendix~\ref{sec:kt}). 

\section{Application to the Hubbard model} \label{sec:hubbard}

To put some numbers on the sort of improvements in sample-complexity that the  results of Section~\ref{sec:results} can imply relative to Eq.~\eqref{eq:inb}, we consider the task of estimating the energy of a state in the spin-$1/2$ \textit{Hubbard model}~\cite{hubbard1963electron,lieb1989two}. Despite its simplicity as a minimal model of correlated electron dynamics, the Hubbard model is something of a standard benchmark in quantum computing~\cite{cade2020strategies,stanisic2022observing}. 
For an (open) chain of $L$ sites, we recall   that Hamiltonian of the Hubbard model may be taken to read
\small
\begin{align*}
H\!=\!-t\sum_{i=1}^{L-1}\!\sum_{\sigma=\uparrow,\downarrow}
\!\big(c_{i\sigma}^{\dagger}c_{i+1,\sigma}
\!+\!{\rm h.c.}\big)\!+\!V\sum_{i=1}^{L}\big(n_{i\uparrow}\,-\tfrac12\big)\big(n_{i\downarrow}-\!\tfrac12\big)    ,
\end{align*}
\normalsize
where $c_{i\sigma}^{\dagger}$ and $c_{i\sigma}$ are the fermionic creation and annihilation operators, respectively, for a fermion with spin $\sigma$ on site $i$, and $n_{i\sigma}=c_{i\sigma}^{\dagger}c_{i\sigma}$ is the corresponding number operator. Then, $t$ and $V$, respectively, denote the hoping and on-site interaction strengths.  Our goal is therefore to use classical shadows to estimate the expectation value $\langle H\rangle_\rho=\Tr(H\rho)$ for some unknown state $\rho$ prepared on a quantum computer.

\begin{figure}
    \centering
    \includegraphics[width=1\columnwidth]{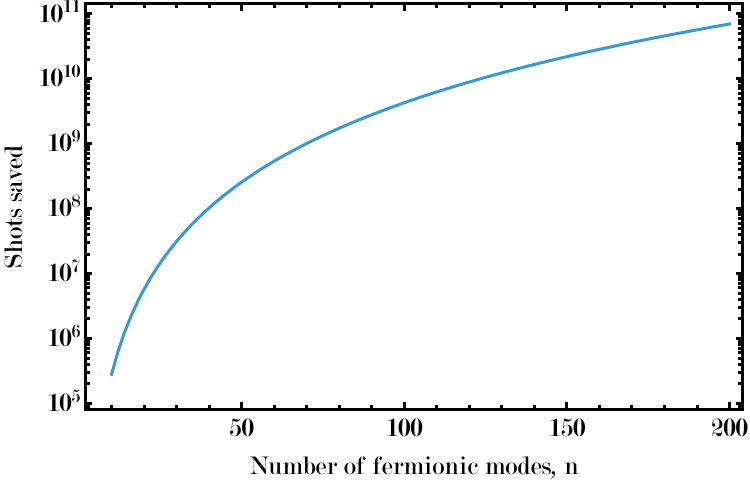}
    \caption{Number of shots saved by the improved matchgate-shadow sample-complexity bound as a function of the number of fermionic modes $n$, for an open Hubbard chain with $t=1$, $V=4$, and target additive precision $\varepsilon=0.1$.}
    \label{fig:scaling}
\end{figure}

To see how this observable fits into the framework described above, let us introduce, for a spin orbital $p=(i,\sigma)$, the Majoranas $x_p=c_p+c_p^\dagger$ and $y_p=-i(c_p-c_p^\dagger)$. Then it is easy to see that
\begin{equation}
n_p-\frac12=\frac{i}{2}x_py_p\,,    
\end{equation}
 and, for $p\neq q$,
\begin{equation}
   c_p^\dagger c_q+c_q^\dagger c_p=\frac{i}{2}(x_py_q-y_px_q)\,. 
\end{equation}
Consequently, we can write $H=H_2+H_4$ with 
\begin{equation}
H_2=-\frac{it}{2}\sum_{i,\sigma}\left(x_{i\sigma}y_{i+1,\sigma}-y_{i\sigma}x_{i+1,\sigma}\right)\in\mathcal L_2
\end{equation}
and
\begin{equation}
H_4=-\frac V4\sum_ix_{i\uparrow}y_{i\uparrow}x_{i\downarrow}y_{i\downarrow}\in\mathcal L_4\,.
\end{equation}
Now, for the energy $h=H/n=h_2+h_4$ per mode, open boundaries give
\begin{align}
\|h_2\|_\infty&=\frac{2|t|}{L}\sum_{j=1}^{\lfloor L/2\rfloor}\cos\!\left(\frac{j\pi}{L+1}\right)\,,\\
\|h_4\|_\infty&=\frac{|V|}{8}\,.
\end{align}

For example, for $n= 100,\ L= 50,\ t= 1$, and $V = 4$, we have $\|h_2\|_\infty\approx 0.63$ and $\|h_4\|_\infty = 1/2 $. 
Using also that $a_{100,1}^{-1}=199$ and $a_{100,2}^{-1}=39203/3$, Eq.~\eqref{eq:mult1} leads to $\Var[\hat h] \leq 6542$. On the other hand, the bound offered by Eq.~\eqref{eq:inb} is at best $(39203/6)^2\approx 4.3\times 10^7$ (this follows from noticing that the vacuum state is an eigenstate of the rescaled Hamiltonian $a_{n,1}^{-1}h_2+a_{n,2}^{-1}h_4$ of eigenvalue $a_{n,2}^{-1}/2$). So, our bound improved by a factor of roughly $10^4$. For a target additive precision of $\varepsilon=0.1$, the inverse  quadratic scaling of Eq.~\eqref{eq:sc} then gives the additional factor of 100 needed to reproduce the numbers quoted in the abstract. More generally, Fig.~\ref{fig:scaling} shows the scaling of how many shots are saved as a function of $n$.

\section{Discussion}
We have  derived asymptotically optimal upper bounds on the  sample-complexity of the learnability of observables of homogeneous fermionic-degree under matchgate shadows. In conjunction with its obvious  theoretical interest,  this result is of practical relevance in the regime of hundreds of qubits already being probed by experiments. Perhaps somewhat surprisingly, our proof of this result is highly technical, and goes significantly beyond the tools needed to prove the analogous bound in the case of a single Majorana monomial~\cite{wan2022matchgate,zhao2021fermionic,west2026classical}. It would be very interesting to find a simpler proof, not least because it is highly unclear how to generalise our current approach for strengthening the generic bounds of Eqs.~\eqref{eq:2nb} and~\eqref{eq:inb} to other classical shadow protocols of interest. Additionally, such a proof might succeeded in removing the factor of 3/2 present in our bounds, which we suspect is an artefact of our proof strategy. Certainly this factor is known to be unnecessary in the case of an individual Majorana monomial~\cite{wan2022matchgate,zhao2021fermionic,west2026classical}, and more generally  numerical experiments points to the optimal constant being one. \\

Finally, the second source of difficulty when experimentally implementing matchgate shadows on real hardware is the circuit depth required to implement matchgate 3-designs. It has been shown that, on one-dimensional nearest-neighbour architectures, sublinear-depth ensembles of matchgates cannot form matchgate 3-designs~\cite{west2025no,grevink2025will}; more recently it has been shown that on such architectures \textit{no} sublinear-depth ensembles can form matchgate 3-designs~\cite{west2026ambient}. On the other hand, it has been shown yet   more recently that on all-to-all connected architectures, matchgate 3-designs can form in the optimal logarithmic depth, from ensembles whose elements lie in the intersection of the matchgate and Clifford groups~\cite{west2026strong}.
Combined with our results here achieving the optimal sample-complexities, this dramatically increases the prospects of performing useful fermionic shadow tomography on near-term hardware, should that hardware support all-to-all interactions. As this connectivity can be achieved naturally in recently developed trapped ion~\cite{moses2023race,ransford2025helios} and neutral atom~\cite{bluvstein2022quantum,bluvstein2024logical} processors, such systems may allow for experimental implementations of  useful matchgate shadows.

\medskip
\section{Acknowledgments}
The authors acknowledge support by the Laboratory Directed Research and Development (LDRD) program of Los Alamos National Laboratory (LANL) under project number 20260043DR, and by LANL’s ASC Beyond Moore’s Law project. This work was also supported by the Quantum Science Center (QSC), a National Quantum Information Science Research Center of the U.S. Department of Energy (DOE). L.C. acknowledges support from the National Science Foundation Graduate Research Fellowship under Grant No. 2140743. Any opinions, findings, and conclusions or recommendations expressed in this material are those of the author(s) and do not necessarily reflect the views of the National Science Foundation. The key technical ideas of the proofs of Theorem~\ref{thm:homogeneous}  are  due to GPT-5.6. Its arguments were verified by the human authors, and  substantially rewritten in order to meet our standards of clarity.

\bibliography{quantum}

\clearpage
\newpage

\onecolumngrid
\appendix

\makeatletter
\renewcommand*{\theHequation}{\theHsection.\arabic{equation}}
\makeatother



\section{Some combinatorics}\label{sec:matching}
In this appendix we  derive  the elementary combinatorial formulae that characterise the relevant aspects of the behaviour of the second and third moments of the uniform matchgate distribution. The critical fact underlying all of this is that the matchgate Cliffords constitute a matchgate 3-design~\cite{wan2022matchgate,heyraud2024unified}, so that we may work with random matchgate Cliffords instead of Haar random matchgates. Now, a Clifford matchgate $U$ acts on Majoranas as a signed permutation, $U\g_\mu U\ad=\epsilon_\mu\gamma_{\pi(\mu)}$, for some $\epsilon_\mu\in\{+1,-1\}$ and $ \pi\in S_{2n}$. As measurements in the computational basis   are (under our convention Eq.~\eqref{eq:majos}) just  the simultaneous measurement of the $Z_j=-i\gamma_{2j-1}\gamma_{2j},\  j=1,\ldots,n$, a  computational measurement after rotating by the random Clifford matchgate $U$ is simply the simultaneous measurement of the operators $\widetilde{Z}_j=-i\gamma_{\pi(2j-1)}\gamma_{\pi(2j)},\  j=1,\ldots,n$.  Equivalently, for a \textit{perfect matching} $M(\pi)=\bigl\{\{\pi(1),\pi(2)\},\ldots,\{\pi(2n-1),\pi(2n)\}\bigr\}$ of $[2n]$, we measure the (rank one projectors)
\begin{equation}
  \Pi_{M,z}
  :=2^{-n}\prod_{e\in M}(\id+z_eZ_e),
  \qquad z\in\{\pm1\}^{M}\,;
\end{equation}
the corresponding dephasing channel is
\begin{equation}
  \mcd_M(X):=\sum_z\Tr(\Pi_{M,z}X)\Pi_{M,z}.
\end{equation}
As explained in the main text, for Majoranas $\Gamma_S,\Gamma_T\in \mcl_{2k}$,  we are interested in understanding the corresponding \textit{survival probabilities}. That is, for a random perfect matching $M$: what is the probability that  $\mcd_M(\Gamma_S)\neq 0$, and the probability that both  $\mcd_M(\Gamma_S)\neq 0 $ and $\mcd_M(\Gamma_T)\neq 0$?
As also mentioned in the main text, the key observation is that a given monomial $\Gamma_S$ is diagonal in the basis  specified by $M$ (and so survives $\mcd_M$) precisely when it is formed from a product of Majoranas who appear if and only if their matching pair (as specified by $M$) also appears, and is otherwise exactly annihilated. Said another way, $\Gamma_S\in\mcl_{2k}$ is retained   exactly when $S$ is a union of $k$ edges of $M$;  it follows that if $|S\cap T|$ is odd, no matching retains both supports. On the other hand,  if  $|S\cap T|=2j$ for some $j$, then their joint survival requires the matching to pair vertices  in each of $S\cap T,\ S\setminus T,\ T\setminus S$, and $[2n]\setminus(S\cup T)$. As there are $(2n-1)!!$ perfect matchings of $[2n]$ for each $n$, the joint survival  is therefore
\begin{equation}
  p_j=\frac{(2j-1)!!\,(2(k-j)-1)!!^2
    (2(n-2k+j)-1)!!}{(2n-1)!!},
  \qquad 0\leq j\leq k\,,
\end{equation}
where we use the convention $(-1)!!:=1$.  Next, recall from the main text the definition $V_\mu=i\Pi\gamma_\mu$, which when acting by conjugation flips the sign of
$\gamma_\mu$  and fixes every other Majorana, and, for
$A\subseteq[2n]$, let $\mathcal U_A$ be the composition of these
commuting sign-flip conjugations. We recall also the \textit{radial shell channel}
\begin{equation}
  \mcr_j:=\binom {2n}{j}^{-1}\sum_{\substack{A\subseteq[2n]\\|A|=j}}\mathcal U_A,\qquad 0\leq j\leq 2n.
\end{equation}
Every $\mcr_j$ is manifestly unital and completely positive.  We adopt the  convention that the Krawtchouk polynomials (not normalised and normalised, respectively) take the form
\begin{align}
  K_s^{(2n)}(j)&=\sum_{\ell=0}^s(-1)^\ell\binom j\ell\binom{2n-j}{s-\ell},\qquad\kappa_s^{(2n)}(j):=\frac{K_s^{(2n)}(j)}{\binom {2n}s}\,.
\end{align}
As the sign acquired by some $\Gamma_S$ under $\mathcal U_A$ is $(-1)^{|A\cap S|}$,  averaging over all $j$-subsets and using  the readily verified identity $(-1)^j\kappa_s^{(2n)}(j)=\kappa_{2n-s}^{(2n)}(j)$ gives
\begin{equation}
  \mcr_j(\Gamma_S)=\kappa_{|S|}^{(2n)}(j)\Gamma_S\,.
\end{equation}

\section{Proof of main results}~\label{sec:kt}

\noindent
In this appendix we fill in the remaining details of the proof of Theorem~\ref{thm:homogeneous}, which we recall to state
\homogeneousbound*
\begin{proof}
The proof sketch of Theorem~\ref{thm:homogeneous} of the main text omitted two details: the derivation of Eq.~\eqref{eq:prob}, and the construction of a function $w:[N]\to\mbr$ (where for notational convenience we introduce $N=2n$) satisfying both $\sum_{j\,:\,w_j>0} w_j\leq3/2$ and 
\begin{equation}
 \widehat w(4r)= \sum_{j=0}^{N}w_j\kappa_{4r}^{(N)}(j) = \frac{\binom{k}{r}\binom{n-k}{r}}{\binom{2k}{2r}\binom{2(n-k)}{2r}}\,.\label{eq:qr}
\end{equation}
That is, the  Krawtchouk transform of $w$ evaluated at the values $4r$ gives exactly the sequence $q_r$ of Eq.~\eqref{eq:sk}. The former omission was remedied in Appendix~\ref{sec:matching}; for the latter, we will need a few identities and lemmas. First, recall that the Krawtchouk generating function and orthogonality relations are given respectively by
\begin{align}
  (1+y)^{N-j}(1-y)^j&=\sum_{s=0}^NK_s^{(N)}(j)y^s,\label{eq:kgen}\\
  \sum_{j=0}^N\pi_N(j)K_s^{(N)}(j)K_t^{(N)}(j)&=\binom Ns\delta_{s,t}\,,\label{eq:ortho}
\end{align}
where $  \pi_N(j):=2^{-N}\binom Nj$ is the measure with respect to which the Krawtchouk polynomials are orthogonal. We will sometimes find it convenient to use the alternate normalisations   $\kappa_s^{(N)}:=K_s^{(N)}/\binom{N}{s}$, and   $\phi_s(j) := K_s^{(N)}(j)/\sqrt{\binom Ns}$. The Krawtchouk polynomials satisfy a useful recurrence relation; with respect, say, to the $\phi_s$, it reads 
\begin{equation}
  (N-2j)\phi_s(j)=a_{s-1}\phi_{s-1}(j)+a_s\phi_{s+1}(j)\,,
  \label{eq:threerec}
\end{equation}
where $a_s = \sqrt{(s+1)(N-s)}$. 
Eq.~\eqref{eq:threerec} follows readily from the generating function Eq.~\eqref{eq:kgen}; it for example also follows that $K_s^{(N)}(j)=[z^s](1+z)^{N-j}(1-z)^j$, where $[z^s]f(z)$ denotes ``the coefficient of $z^s$ in the polynomial $f(z)$.'' \\

\noindent
Next, let us introduce the function
\begin{equation}\label{eq:ntau}
  \nu_\tau(j):=\binom Nj
  \left(\frac{1+\tau}{2}\right)^{N-j}
  \left(\frac{1-\tau}{2}\right)^j.
\end{equation}
For our purposes, the useful thing about $\nu_\tau$ is that it  has a very nice Krawtchouk transform:
\begin{align}
\widehat{\nu_\tau}(s)&=\frac{1}{\binom Ns}\sum_{j=0}^N\binom Nj \left(\frac{1+\tau}{2}\right)^{N-j}\left(\frac{1-\tau}{2}\right)^j K_s^{(N)}(j)\\
&=\frac{1}{\binom Ns}[z^s]
\sum_{j=0}^N\binom Nj \left(\left(\frac{1+\tau}{2}\right)(1+z)\right)^{N-j} \left(\left(\frac{1-\tau}{2}\right)(1-z)\right)^j\\
&=\frac{1}{\binom Ns}[z^s]\left(\frac{(1+\tau)(1+z)+(1-\tau)(1-z)}{2}\right)^N\\
&=\frac{1}{\binom Ns}[z^s](1+\tau z)^N\\
&=\tau^s.\label{eq:ntt}
\end{align}
We will put $\nu_\tau$ to use in a moment. First, though, we notice  that $q_r$ has a useful probabilistic interpretation. Indeed, for some positive integer $p$, consider a random variable $T_p$ distributed according to  the \textit{beta-prime distribution}~\cite{cordeiro2012mcdonald}, which for $t>0$ has density
\begin{equation}\label{eq:bpd}
  f_p(t)=\frac{t^{-1/2}(1+t)^{-p-1}}{\mathrm B(1/2,p+1/2)}\,,
\end{equation}
where $B$ denotes the beta function, and, critically, for $0\leq r\leq p$, moments~\cite{cordeiro2012mcdonald}
\begin{align}
  \mbe[T_p^r]=\frac{\Gamma(r+1/2)\Gamma(p-r+1/2)}
  {\Gamma(1/2)\Gamma(p+1/2)}=\frac{\binom pr}{\binom{2p}{2r}}\,
  .
\end{align}
Taking   $T_k$ and $T_m$ (with $m=n-k$) to be independent such variables, then, we see from Eq.~\eqref{eq:qr} that if $T=T_kT_m$ we have $q_r=\mbe[T^r]$ whenever $  0\leq r\leq k$. At this point, we are nearly done; indeed consider taking, with $T$ so-distributed, $w_j \underset{?}{=} \mbe_T[\nu_{T^{1/4}}(j)]$. By the discussion of the last few paragraphs, we would then by linearity have
\begin{equation}
    \widehat{w}(4r) = \mbe_T[\widehat{\nu_{T^{1/4}}}(4r)]= \mbe_T[ (T^{1/4})^{4r}]=\mbe_T[ T^{r}] = q_r\,.
\end{equation}
The difficulty is that in order to prove Theorem~\ref{thm:homogeneous} we also need tight control over $\|w\|_{\rm TV}$, which we will not get from this choice of $w$. In fact, under our current ``definition'', $w$ is not even well-defined: it follows from Eq.~\eqref{eq:bpd} that $\mbe[T^r_p]$ diverges for $r>p+1/2$, but for $T\to\infty$ we have from Eq.~\eqref{eq:ntau} that $\nu_{T^{1/4}}\sim 2^{-N}\binom NJ (-1)^j T^{N/4}$, which will therefore diverge if $N/4=n/2>k+1/2$, which is generically true. So, we need to do a little bit of work to coax $w$ into a suitably normalised form (without inadvertently destroying  its nice transformation properties). To that end, we have a technical lemma (the proof of which we shall for the sake of narrative continuity defer to the next Appendix):
\begin{restatable}{lemma}{lemtwo}\label{lem:2}
  For $0\leq d\leq N/2$, let $B_{N,d}:={\binom Nd}/{\binom{2d}d}$. Then, for every $0\leq y\leq1$, there is a function $\rho_y:[N]\to\mbr^+$  such that, for all $2d\leq s\leq N$, we have $\widehat{\rho_y}(s)=y^{s-2d}$, and $\sum_{j=0}^N\rho_y(j)\leq B_{N,d}$ (at  $y=0$, we  take $y^{s-2d}=\delta_{s,2d}$).
\end{restatable}
Now, with $d=n-2k$, let $\rho_y$ be the function supplied by the lemma, and consider $\mu_y(j) = (-1)^j\rho_y(j)$.
Using again the identity $(-1)^j\kappa_s^{(N)}(j)=\kappa_{N-s}^{(N)}(j)$, we have that $\widehat{\mu_y}(4r)=\widehat{\rho_y}(2n-4r)=y^{2n-4r-2d}=y^{4k-4r}$, and subsequently that $\widehat{\mu_{T^{-1/4}}}(4r)=T^{r-k}$. Consider then, with $\boldsymbol{1}$ the indicator function,   the updated ansatz 
\begin{equation}
    w(j)= \mbe_T\big[\boldsymbol{1}_{T\leq 1}\nu_{T^{1/4}}(j)+\boldsymbol{1}_{T> 1}T^k\mu_{T^{-1/4}}(j)\big],
\end{equation}
which by construction satisfies $\widehat{w}(4r)=q_r$, and furthermore that 
\begin{equation}\label{eq:magsum}
    \sum_j \lvert w(j)\rvert \leq \Pr(T\leq 1) + B_{N,d}\mbe_T[T^k\boldsymbol{1}_{T> 1}]\leq 1 + B_{N,d}q_k \leq 2\,,
\end{equation}
where $B_{N,d}q_k\leq 1$ can be verified by direct algebra.   Finally, $\sum_j w_j = \widehat{w}(0)=q_0=1$; combined with Eq.~\eqref{eq:magsum}   it immediately follows that $\sum_{j\,:\,w_j>0} w_j\leq 3/2$.  

\end{proof}

\noindent
Next we prove Corollary~\ref{crl:mult} (restated for convenience):
\crlmult*
\begin{proof}
As mentioned in the main text, this is essentially an application of the triangle inequality. Indeed, for a given unknown target state $\rho$, the total (centered) estimator decomposes into contributions from each component:
\begin{equation}
\hat o-\mathbb E_\rho[\hat o] = \sum_{k=1}^K\big[ \hat o_{2k}-\mathbb E_\rho[\hat o_{2k}]\big]
\end{equation}
Now, the variance is just the squared 2-norm of this variable (weighted by the probabilities), i.e.
\begin{align}
\Var_\rho [\hat{o}] &= \mbe_\rho [(\hat o-\mathbb E_\rho[\hat o])^2]\\
&= \sum_{U,w} \tr[\rho U\ad \Pi_w U]  (\hat o(U,w)-\mathbb E_\rho[\hat o])^2\,
\end{align}
which is the 2-norm (squared) of the vector $\boldsymbol{v}$ with components (addressed by pairs $(U,w)$) given by 
\begin{equation}
    \boldsymbol{v}_{(U,w)} = \sqrt{\tr[\rho U\ad \Pi_w U]}(\hat o(U,w)-\mathbb E_\rho[\hat o])\,.
    \label{eq:def_v}
\end{equation}
We can therefore apply the triangle inequality (with respect to this 2-norm) to obtain
\begin{equation}
   \sqrt{ \Var_\rho [\hat{o}]} = \|\boldsymbol{v}\|_2 \leq \sum_{k=1}^K  \|\boldsymbol{v}_k\|_2 =\sum_{k=1}^K \sqrt{ \Var_\rho [\hat{o}_{2k}]} \leq \sum_{k=1}^K \sqrt{\frac{3}{2a_{n,k}}}\|O_{2k}\|_\infty\,,
\end{equation}
where $\boldsymbol{v}_k$ is defined analogously to $\boldsymbol{v}$ in Eq.~\eqref{eq:def_v}, with $\hat{o}$ replaced by $\hat{o}_{2k}$.
This immediately yields Eq.~\eqref{eq:mult1}. Eq.~\eqref{eq:mult2} then follows from the generic inequality $(\sum_{k=1}^K x_k)^2\leq K\sum_{k=1}^K x_k^2$.

\end{proof}

\section{Further minutiae}
\noindent
In this appendix we supply the proof of Lemma~\ref{lem:2}, which we restate for convenience:
\lemtwo*
\begin{proof}
  We will work through this in a case-by-case fashion. First, if $d=0$ we can by Eq.~\eqref{eq:ntt} satisfy the transform condition by simply taking $\rho_y = \nu_y$; by Eq.~\eqref{eq:ntau} $\nu_y$ is a probability distribution, and therefore further satisfies $\sum_{j=0}^N\rho_y(j)=1=B_{N,0}$. So, let us take $d\geq 1$. We next take care of the special cases $y=0,1$. If $y=0$, we want a function $\rho_0$ such that $\widehat{\rho_y}(s)=\delta_{s,2d}$; let us show that  
  \begin{equation}
      \rho_0(j) = \pi_N(j) \frac{\big(K_d^{(N)}(j)\big)^2}{\binom{2d}{d}}
  \end{equation}
  suffices. First, its Krawtchouk transform is certainly proportional to $\delta_{s,2d}$,
  as $K_d^{(N)}(j)^2$ is a polynomial of degree $2d$ in $j$, so that by orthogonality its inner product with $\kappa_s^{(N)}$ must vanish for $s> 2d$. Concretely, it follows from Eq.~\eqref{eq:kgen} that the leading coefficient of $K_s^{(N)}$ is ${(-2)^s}/{s!}$ (so that  the leading coefficients of $(K_d^{(N)})^2$ and $K_{2d}^{(N)}$ and are respectively ${4^d}/{(d!)^2}$ and ${4^d}/{(2d)!}$); that is, 
   \begin{align}
      \widehat{\rho_0}(s)&=\sum_j \Bigg( \pi_N(j) \frac{\big(K_d^{(N)}(j)\big)^2}{\binom{2d}{d}}\Bigg)\kappa_s^{(N)}(j) \\
      &=\frac{\delta_{s\leq 2d}}{\binom{2d}{d}}\sum_j \Bigg( \pi_N(j) \big(K_d^{(N)}(j)\big)^2\Bigg)\kappa_{2d}^{(N)}(j) \\
      &= \frac{\delta_{s\leq 2d}}{\binom{2d}{d}} \sum_j  \pi_N(j) \Big( \frac{4^d (2d)!}{4^dd!^2}K_{2d}^{(N)}(j)
+\text{lower-degree terms}\Big)\kappa_{2d}^{(N)}(j)\\
&= \delta_{s\leq 2d}\,,
  \end{align}
  where we have used Eq.~\eqref{eq:ortho}. Recalling the assumption  $2d\leq s\leq N$ of the lemma, we conclude that in the relevant range the transform is indeed $\widehat{\rho_0}(s)=\delta_{s,2d}$. The orthogonality relations of Eq.~\eqref{eq:ortho} also immediately yield $\sum_j \rho_0(j) = \binom{N}{d}/\binom{2d}{d}=B_{N,d}$,  we also satisfy the normalisation demand of the lemma.  The case $y=1$ is even simpler. In this instance we seek a function $\rho_1$ such that $\widehat{\rho_1}(s)=1^{s-2d}=1$; as it follows from Eq.~\eqref{eq:kgen} that $\kappa_s^{(N)}(0)=1$ for all $s$, we can simply take $\rho_1(j)=\delta_{j,0}$. This choice also  evidently satisfies the mass   requirement, as $\sum_j \rho_1(j) = 1\leq B_{N,d} $, as $N\geq 2d$. This concludes the analysis of the easy edge cases.\\
  
  So, let us take $d\geq 1$ and $0<y<1$. We will need a few elementary results from the theory of orthogonal polynomials. First, with  $\xi_1,\ldots,\xi_d$  the roots of $K_d^{(N)}$, there exist so-called  \textit{Gaussian-quadrature weights} $\omega_1,\ldots,\omega_d>0$ of $\pi_N$ (the measure, recall, with respect to which the Krawtchouk polynomials are orthogonal) such that we have for any polynomial $P$ with $\deg P\leq2d-1$ that~\cite{gautschi2004orthogonal,golub1969calculation}
\begin{equation}\label{eq:interp}
  \sum_{j=0}^N\pi_N(j)P(j)
  =\sum_{i=1}^d\omega_iP(\xi_i)\,.
\end{equation}
Now, introducing the function  
\begin{equation}\label{eq:gy}
  G_y(t)=\frac{\nu_y(t)}{\pi_N(t)}=(1+y)^{N-t}(1-y)^t\,,
\end{equation}
which we note to satisfy
\begin{equation}
 \sum_{j=0}^N\pi_N(j)G_y(j)=1\,, 
  \label{eq:gnorm}
\end{equation}
and letting $H_y$ be the unique polynomial of degree at most $2d-1$ that agrees with $G_y$ and its first derivative at every $\xi_i$. The next technical result that we need is the \textit{Hermite remainder formula}, which  gives~\cite{gautschi2004orthogonal}
\begin{equation}
  G_y(t)-H_y(t)=\frac{G_y^{(2d)}(\zeta_t)}{(2d)!}\prod_{i=1}^d(t-\xi_i)^2\geq0\,,\label{eq:herm_rem}
\end{equation}
for $0\leq t\leq N$, and some $\zeta_t$ in the (smallest) interval containing $t\,\cup\,\{\xi_i\}_{i=1}^d$. The non-negativity of $G_y-H_y$ asserted in Eq.~\eqref{eq:herm_rem} follows from a direct calculation which shows that   the even derivative
$G_y^{(2d)}$ is positive.  Hence we have that 
\begin{equation}
  \rho_y(j):=\pi_N(j)y^{-2d}\bigl(G_y(j)-H_y(j)\bigr)\,,
\end{equation}
is positive. Now, since $H_y$ has degree less than $2d$, it is orthogonal to the Krawtchouk polynomials of degree $s\geq2d$, so that
Eq.\eqref{eq:gy} shows, for $s\geq2d$, that
\begin{align}
  \widehat{\rho_y}(s)&=\sum_j\rho_y(j)\kappa_s^{(N)}(j)\\
  &=y^{-2d}\sum_j\pi_N(j)G_y(j)\kappa_s^{(N)}(j)\\
  &=y^{-2d}\sum_j\nu_y(j)\kappa_s^{(N)}(j)\\
  &=y^{-2d}\widehat{\nu_y}(s)\\
  &=y^{s-2d}\,,
\end{align}
which is exactly how we want the Krawtchouk transform of $\rho_y$ to behave. The remaining demand of the lemma it to control its mass (i.e. $\sum_j\rho_y(j)$), and is the task to which we now turn. \\

\noindent
Now, from Eqs.~\eqref{eq:interp} and~\eqref{eq:gy}, we have the preliminary simplification
\begin{align}
\sum_j\rho_y(j)&=y^{-2d}\Bigl(\sum_{j=0}^N\pi_N(j)G_y(j)-\sum_{j=0}^N\pi_N(j)H_y(j)\Bigr)\\
&= y^{-2d}\Bigl(1-\sum_{i=1}^d\omega_iH_y(\xi_i)\Bigr)\\
  &= y^{-2d}\Bigl(1-\sum_{i=1}^d\omega_iG_y(\xi_i)\Bigr)\,,\label{eq:mass}
\end{align}
  where we have used  Eqs.~\eqref{eq:gnorm} and~\eqref{eq:interp}, and that by construction $H_y$ and $G_y$ agree at every $\xi_i$. 
  We also notice that the substitution $y=\tanh u$ in Eq.~\eqref{eq:gy} leads via elementary means to 
  \begin{equation}\label{eq:gexp}
  G_y(j)=(\cosh u)^{-N}e^{u(N-2j)}\,.
\end{equation}

The previous result turns out to be a very important reformulation. In particular, upon seeing the factor $e^{u(N-2j)}$,  one cannot help but recall the recurrence relation Eq.~\eqref{eq:threerec}; indeed this will turn out to be the key to further progress. Briefly taking something of a step back, consider the Hilbert space $\mcf=\ell^2([N], \pi_N)$, i.e. the space of functions from $[N]$ to $\mbc$, equipped with the inner product $\langle f|g\rangle_{\pi_N}  = \sum_j \pi_N(j)\overline{f(j)}g(j)$. Evidently the functions $\phi_s$ form an orthonormal basis of $\mcf$,
  and Eq.~\eqref{eq:threerec} is the statement that the matrix representation (with respect to this basis) of the linear operator $A$ which represents multiplication by
  $x_j:=N-2j$ takes the form 
  \begin{equation}
    A = \begin{pmatrix}  
      0 & a_0 & 0 & \cdots\vspace{1.5mm}\\
      a_0 & 0 & a_1 & \cdots\\
      0 & a_1 & 0 & \ddots\\
      \vdots & \vdots & \ddots & \ddots\\
     \end{pmatrix} \,,
    \label{eq:a}
  \end{equation}
where $a_s = \sqrt{(s+1)(N-s)}$. 
One could think of $A$ as a  (weighted) adjacency matrix of the chain $(0,1,\ldots,N)$, where the edge $\{s,s+1\}$ has weight $a_s$.
Now, consider the family (indexed by $L\geq 0$) of vectors $m^{(L)}\in\mbr^{N+1}$ whose $s$\textsuperscript{th} coordinates are given by  $m_s^{(L)} = \sum_{j=0}^N \pi_N(j)(x_j)^L\phi_s(j)$ for $0\leq s\leq N$. For example,
$m^{(0)}_s =  \sum_{j=0}^N \pi_N(j)\phi_s(j)=\langle \phi_s|\phi_0\rangle_{\pi_N}=\delta_{s,0}$, so that $m^{(0)}=e_0$ as a vector in $\mbr^{N+1}$. 
Now, from Eq.~\eqref{eq:threerec}, we have
\begin{equation}
  m_s^{(L+1)} =   \sum_{j=0}^N \pi_N(j)(x_j)^{L+1}\phi_s(j) = a_{s-1}m_{s-1}^{(L)}+a_{s}m_{s+1}^{(L)}\,,
\end{equation}
so that, in vector form, $m^{(L)}=Am^{(L-1)}=A^Le_0$; taking the 0\textsuperscript{th} coordinate (and using $\phi_0=1)$ then gives
\begin{equation}
  (A^L)_{00} = (m^{(L)})_0=\sum_{j=0}^N \pi_N(j)x_j^{L}= \sum_{j=0}^N \pi_N(j)(N-2j)^{L}\,.
\end{equation}
By linearity and the power series of the exponential we find
\begin{equation}
  (e^{uA})_{00} = \sum_{j=0}^N \pi_N(j)e^{u(N-2j)}\,;\label{eq:expzz}
\end{equation}
in conjunction with Eqs.~\eqref{eq:gnorm} and~\eqref{eq:gexp}, we conclude that $  (e^{uA})_{00} =\cosh^N(u)$. 

Far from clear though it may be at this point, the utility of the previous equality will turn out to be that we can substitute $(e^{uA})_{00}\cosh^{-N}(u)$ for the ``1'' on the right hand side of Eq.~\eqref{eq:mass} (recall that our present goal  is to show  that the right hand side of Eq.~\eqref{eq:mass} is upper bounded by $B_{N,d}$). 
Thus, let us turn to the other term on the right hand side of Eq.~\eqref{eq:mass}, namely $\sum_{i=1}^d\omega_iG_y(\xi_i)$. We will play a similar game to the one that led to Eq.~\eqref{eq:expzz}, but this time considering the vector $q^{(L)}\in\mbr^d$ with components $q_s^{(L)}  = \sum_{i=1}^d\omega_i(N-2\xi_i)^L\phi_s(\xi_i)$, for $0\leq s\leq d-1$. To begin, we note that as  $\phi_s$  has degree $s\leq d-1$, Eq.~\eqref{eq:interp} gives
\begin{equation}
  q_s^{(0)}  = \sum_{i=1}^d\omega_i\phi_s(\xi_i) = \sum_{j=0}^N\pi_N(j)\phi_s(j)  =\delta_{s,0}\,,
\end{equation}
so that $q^{(0)}  = e_0$. Now, evaluating Eq.~\eqref{eq:threerec} at $j=\xi_i$ gives $q_s^{(L+1)} = a_{s-1}q_{s-1}^{(L)}+a_{s}q_{s+1}^{(L)}$; the key difference from the previous calculation is that $\phi_d(\xi_i)=0$, so that the recursion restricts exactly to the indices $0,\ldots,d-1$. In other words, and denoting  by $A_{<d}$ the $d\times d$ ``top left hand corner'' of $A$, we have $q^{(L)}=A_{<d}^L e_0$. Taking again 0\textsuperscript{th} coordinates leads as before to 
\begin{equation}
  (e^{uA_{<d}})_{00} = \sum_{i=1}^d \omega_ie^{u(N-2\xi_i)}\,.
\end{equation}
It now follows immediately from Eq.~\eqref{eq:gexp} that we have
\begin{equation}
  1-\sum_{i=1}^d\omega_iG_y(\xi_i) = (\cosh u)^{-N}\big((\cosh u)^N -   (e^{uA_{<d}})_{00} \big)= (\cosh u)^{-N}\big((e^{uA})_{00} -   (e^{uA_{<d}})_{00} \big)\,.\label{eq:expsub}
\end{equation}
The nice thing about formulating things in this manner is that, with 
one thinking of $A$ as a (weighted) adjacency matrix of the chain $(0,1,\ldots,N)$, where the edge $\{s,s+1\}$ has weight $a_s$, the difference  $(e^{uA})_{00} -   (e^{uA_{<d}})_{00} $ is exactly the weight of the walks from 0 to 0 that visit the  vertex $d$ (possibly more than once). These walks are themselves in bijection with the walks from $0$ to $2d$. One natural map, which we will call $\Upsilon$, is obtained by splitting a walk from $0$ to $0$ at the first timestep at which it reaches the vertex $d$ into subwalks $P$ (from $0$ to $d$) and $Q$ (from $d$ back to $0$). We then construct a walk from $0$ to $2d$ by traversing $Q$ in reverse, followed by the walk obtained by both reversing and reflecting $P$ about $d$.  Essentially the point is that on the one hand a walk from 0 to 0 which reaches the vertex $d$ must hit $d$ at some minimal time (and maybe again later), and a walk from 0 to $2d$ must hit $d$ at some maximal time (and then never again); our bijection is constructed to match up the otherwise unconstrained subwalks in the two cases. The utility of all this is that we can bound the total weight of the walks from 0 to 0 which visit $d$ in terms of the walk from zero to $2d$, which will turn out to be a little easier to calculate. \\

Now, if $m_\ell(P)$ is the number of times the walk $P$ from 0 to $d$ traverses the edge
$\{\ell,\ell+1\}$, then $m_\ell(P)$ is odd and at least one for every
$\ell<d$.  Moreover, the total  weight $\text{wt}(P)=\prod_{\ell=0}^{d-1}a_\ell^{m_\ell(P)}$ of the walk is less than the weight of the translated (and, irrelevantly for the purpose of calculating the weight, reflected) walk from  $d$ to $2d$, as $  a_{2d-1-\ell}^2-a_\ell^2=(N-2d)(2d-2\ell-1)\geq0$. It follows that
\begin{equation}
    \frac{\text{wt}(PQ)}{\text{wt}(\Upsilon(PQ))} = \prod_{\ell=0}^{d-1}\frac{a_\ell^{m_\ell(P)}}{a_{2d-1-\ell}^{m_\ell(P)}} \leq \prod_{\ell=0}^{d-1}\frac{a_\ell }{a_{2d-1-\ell} }=\prod_{\ell=0}^{d-1}\sqrt{\frac{(\ell+1)(N-\ell)}{(2d-\ell)(N-2d+\ell+1) }} = \frac{B_{N,d}}{\sqrt{\binom{N}{2d}}}\,,
\end{equation}
from which the  discussion of the previous paragraph implies
\begin{equation}
    (e^{uA})_{0,0} -   (e^{uA_{<d}})_{0,0} \leq \frac{(e^{uA})_{0,2d}B_{N,d}}{\sqrt{\binom{N}{2d}}}\,,\label{eq:expbound}
\end{equation}
with it remaining to evaluate $(e^{uA})_{0,2d}$. To get there, we have a cute   reformulation in terms of the so-called Dicke states $\{|0\rangle,|1\rangle,\ldots,|N\rangle\}$, defined in terms of the $N$-qubit computational basis states $\boldsymbol{ x}$ as
\begin{equation}
    \ket s = 
\frac{1}{\sqrt{\binom Ns}}
\sum_{|\boldsymbol{ x}| = s}|\boldsymbol{ x}\rangle \,.
\end{equation}
Indeed, consider the symmetrised Pauli $X$ operator, $X_{\rm sym}=\sum_{i=1}^N X_i$. Evidently $X_{\rm sym}$ maps a Dicke state $\ket s$ to some linear combination of $\ket{s\pm 1}$; in fact it is not difficult to see that we have $X_{\rm sym}\ket s = a_{s-1}\ket{s-1}+a_s\ket{s+1}$,  which  is exactly the action of $A$ from Eq.~\eqref{eq:a}! Now, as the $X_i$ commute, we have $e^{uX_{\rm sym}} = \bigotimes_{i=1}^N(\cosh (u)\id + \sinh(u)X_i)$, so that 
\begin{equation}
    (e^{uA})_{0,s}=\sqrt{\binom{N}{s}}\sinh^s(u)\cosh^{N-s}(u)\,.
\end{equation}
Combining Eqs.~\eqref{eq:mass},~\eqref{eq:expsub} and~\eqref{eq:expbound} at $s=2d$ then yields
\begin{equation}
    \sum_j\rho_y(j) \leq B_{N,d} y^{-2d} \cosh^{-N} (u)\sinh^{2d}(u)\cosh^{N-2d}(u) = B_{N,d} y^{-2d}\tanh^{2d}(u)=B_{N,d}\,,
\end{equation}
because $y=\tanh u$. This concludes the proof of the lemma.

\end{proof}

\end{document}